\documentclass[conference]{IEEEtran} 
\usepackage[T1]{fontenc}
\usepackage{graphicx}
\usepackage[cmex10]{amsmath}
\usepackage{amsthm} 
\usepackage{amsfonts}
\usepackage{amssymb} 
\usepackage{graphicx}%
\usepackage{color}%
\usepackage{algorithmic}
 \usepackage{algorithm}
\newtheorem{theorem}{Theorem}

\usepackage[english]{babel}
\usepackage{blindtext}
\usepackage{epsfig}
\usepackage{framed}
\newtheorem{remark}{Remark}
\usepackage[table,usenames,dvipsnames]{xcolor}

\title{On the Deterministic Sum-Capacity of the Multiple Access Wiretap Channel}

\author{
\IEEEauthorblockN{Rick Fritschek}
\IEEEauthorblockA{Heisenberg Communications and Information Theory Group\\
    Freie Universit\"at Berlin, \\
    Takustr. 9,
    D--14195 Berlin, Germany\\
    Email: rick.fritschek@fu-berlin.de
}%
\and
\IEEEauthorblockN{Gerhard Wunder}
\IEEEauthorblockA{Heisenberg Communications and Information Theory Group\\
	Freie Universit\"at Berlin, \\
    Takustr. 9,
    D--14195 Berlin, Germany\\
Email: wunder@zedat.fu-berlin.de}

}

\begin{document}

\maketitle
\begin{abstract}
We study a deterministic approximation of the two-user multiple access wiretap channel. This approximation enables results beyond the recently shown $\tfrac{2}{3}$ secure degrees of freedom (s.d.o.f.) for the Gaussian multiple access channel. While the s.d.o.f. were obtained by real interference alignment, our approach uses signal-scale alignment. We show an achievable scheme which is independent of the rationality of the channel gains. Moreover, our result can differentiate between channel strengths, in particular between both users, and establishes a secrecy rate dependent on this difference. We can show that the resulting achievable secrecy rate tends to the s.d.o.f. for vanishing channel gain differences. Moreover, we extend the s.d.o.f. bound towards a general bound for varying channel strengths and show that our achievable scheme reaches the bound for certain channel gain parameters. We believe that our analysis is the first step towards a constant-gap analysis of the Gaussian multiple access wiretap channel.
\end{abstract}

\section{Introduction}
In this paper we study the secure capacity of an approximation of the Gaussian multiple access wiretap channel. The wiretap channel was first proposed by Wyner in \cite{Wyner75}, and solved in its degraded version. The result was later extended to the general wiretap channel by Csiszar and K\"orner in \cite{CsizarKoerner}. Moreover, the Gaussian equivalent was studied by Leung-Yan-Cheon and Hellman in \cite{Hellman}. The wiretap channel and its modified version served as an archetypical channel for physical-layer security investigations. However, in recent years, the network nature of communication, i.e. support of multiple users, becomes increasingly important. A straight-forward extension of the wiretap channel to multiple users was done in \cite{TekinYenerGMAC-WT}, where the Gaussian multiple access wiretap channel (GMAC-WT) was introduced. A general solution for the secure capacity of this multi-user wiretap set-up was out of reach and investigations focused on the secure degrees of freedom (s.d.of.) of these networks. Degrees of freedom are used to gain insights into the scaling behaviour of multi-user channels. They measure the capacity of the network, normalized by the single-link capacity, as power goes to infinity. This also means that the d.o.f. provide an asymptotic view on the problem at hand. This simplifies the analysis and enables asymptotic solutions of channel models where no finite power capacity results could be found. An example for a technique which yields d.o.f. results is real interference alignment. It uses integer lattice transmit constellations which are scaled such that alignment can be achieved. The intended messages are recovered by minimum-distance decoding and the error probability is bounded by usage of the Khintchine-Groshev theorem of Diophantine approximation theory. The disadvantage of the method is that these results only hold for almost all channel gains. This is unsatisfying for secrecy purposes, since it leaves an infinite amount of cases where the schemes do not work, e.g. rational channel gains. Moreover, secrecy should not depend on the accuracy of channel measurements. Real interference alignment is part of a broader class of interference alignment strategies. Interference alignment (IA) was introduced in \cite{CadambeJafarIA} and \cite{Maddah-AliKhandi-IA}, among others, and its main idea is to design signals such that the caused interference overlaps(aligns) and therefore uses less signal dimensions. The resulting interference-free signal dimensions can be used for communication. IA methods can be divided into two categories, namely the vector-space alignment approach and the signal-scale alignment approach \cite{Niesen-Ali}. 
The former uses the classical signalling dimensions time, frequency and multiple-antennas for the alignment, while the latter uses the signal strength for alignment. Real interference alignment and signal-strength deterministic models are examples for signal-scale alignment. Signal-strength deterministic models are based on an approximation of the Gaussian channel. An example for such an approximation is the linear deterministic model (LDM), introduced by Avestimehr et al. in\cite{Avestimehr2007}. It is based on a binary expansion of the transmit signal, and an approximation of the channel gain to powers of two. The resulting binary expansion gets truncated at the noise level which yields a noise-free binary signal vector and makes the model deterministic. It has been shown that various Gaussian channels (i.e. 
\cite{Bresler2008}, \cite{Bresler2010}, \cite{Suvarup2011}, \cite{FW14b}) can be approximated by the LDM such that the deterministic capacity is within a constant bit-gap of the Gaussian channel. Moreover, layered lattice coding schemes can be used to transfer the achievable scheme to the Gaussian model.

{\bf Previous work and Contributions:} 
Previous results on the  multiple access wiretap channel include the sum secure d.o.f. of $\tfrac{K(K-1)}{K(K-1)+1}$ for the K-user case in \cite{XieUlukusOneHop}. There, the authors used a combination of real interference alignment together with cooperative jamming \cite{TekinYenerCoopJam}. The idea is that the users can allocate a small part of the signalling dimensions with uniform distributed random bits. Those random bits are send such that they occupy a small space at the legitimate receiver, while overlapping with the signals at the eavesdropper. The crypto-lemma shows, that the eavesdropper cannot recover the signal without knowledge of the random bits, hence securing communication. The next step is to transition from the d.o.f. result, to a constant-gap capacity result. 
In \cite{FW16_ISIT} the previously known $\tfrac{1}{2}$ d.o.f. result of the wiretap channel with a helper \cite{XieUlukusOneHop} was extended towards a constant-gap result for certain channel gains. This was achieved by approximating the problem with the LDM and then transferring the results to the Gaussian model. We will use the same strategy for the multiple access wiretap channel, i.e. restrict the problem to the two user case and approximate it by the LDM. We develop a novel communication scheme which achieves the $\tfrac{2}{3}$ d.o.f. result as baseline and extends to a generalized s.d.o.f. characterization. Moreover, we transfer bounding techniques of \cite{XieUlukusOneHop} to the deterministic setting and extend them towards a general bound, which also includes asymmetrical channel gains.

\begin{figure}
\centering
\includegraphics[scale=0.82]{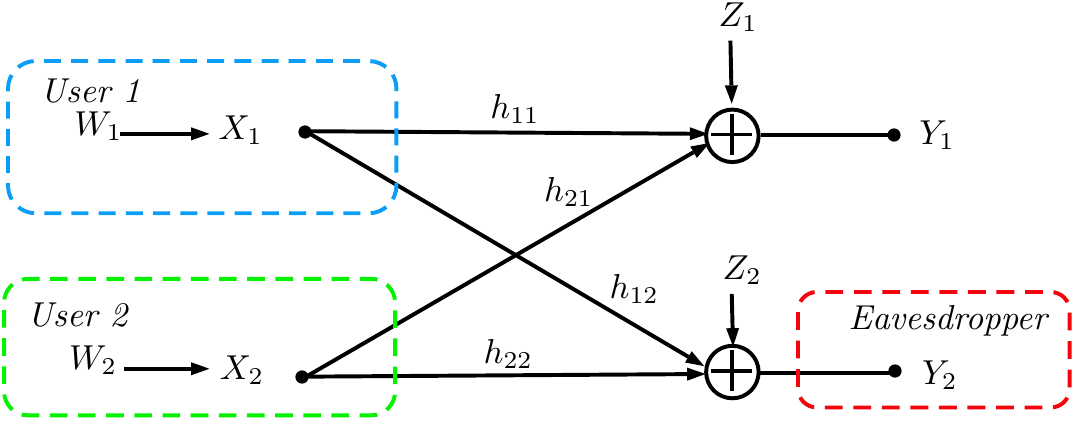}
\caption{The Gaussian multiple access Wiretap channel.}
\end{figure}
\label{System model}
\section{System Model} 

The Gaussian multiple access wiretap channel is defined as a system consisting of two transmitters and two receivers. Each transmitter $i$ has a channel input $X_i$ which it sends over the channel to the legitimate receiver $Y_1$, whereas the eavesdropper $Y_2$ tries to intercept the messages. The channel itself is modelled with additive white Gaussian noise, $Z_1,Z_2$. The system equations can therefore be written as 
\begin{IEEEeqnarray}{rCl}\label{Gauss_Model}
Y_1&=&h_{11}X_1+h_{21}X_2+Z_1\IEEEyessubnumber\\
Y_2&=&h_{22}X_2+h_{12}X_1+Z_2,\IEEEyessubnumber
\end{IEEEeqnarray}
where the channel inputs satisfy an input power constraint $E\{|X_i|^2 \} \leq P$ for each $i$. Moreover, the Gaussian noise terms are assumed to be independent and zero mean with unit variance, $Z_i \sim \mathcal{CN}(0,1)$. A $(2^{nR_1},2^{nR_2},n)$ code for the multiple access wiretap channel will consist of a message pair $(W_1,W_2)$ uniformly distributed over the message set $[1:2^{nR_1}]\times[1:2^{nR_2}]$ with a decoding and two randomized encoding functions. Encoder $1$ assigns a codeword $X_1^n(w_1)$ to each message $w_1\in [1:2^{nR_1}]$, while the encoder $2$ assigns a codeword $X_2^n(w_2)$ to each message $w_2\in [1:2^{nR_2}]$. The decoder assigns  an estimate $(\hat{w_1},\hat{w_2})\in [1:2^{nR_1}]\times [1:2^{nR_2}]$ to each observation of $Y_1^n$. A rate is said to be achievable if there exist a sequence of $(2^{nR_1},2^{nR_2},n)$ codes, for which the probability of error $P^{(n)}_e=P[(\hat{W_1},\hat{W_2})\neq (W_1,W_2)]$  goes to zero, as $n$ goes to infinity $\lim_{n\rightarrow \infty} P^{(n)}_e =0$. A message $W$ is said to be information-theoretically secure if the eavesdropper cannot reconstruct the message $W$ from the channel observation $Y_2^n$. This means that the uncertainty of the message is almost equal to its entropy, given the channel observation. Considering both messages $W_1,W_2$, we have that
\begin{equation}
\tfrac{1}{n} H(W_1,W_2|Y_2^n) \geq \tfrac{1}{n} H(W_1,W_2)-\epsilon,
\label{security_constraint}
\end{equation}
which leads to $I(W_1,W_2;Y_2^n)\leq \epsilon n$ for any $\epsilon>0$. A secrecy rate $r$ is said to be achievable if it is achievable while obeying the secrecy constraint \eqref{security_constraint}.
\subsection{The Linear Deterministic Model}

We investigate the linear deterministic model (LDM) of the multiple access wiretap channel as simplification of the corresponding Gaussian model. This approximation models a signal in the channel as bit-vector $\mathbf{X}$, which is achieved by a binary expansion of the input signal $X$. The elements of the resulting vector are referred to as bit-levels. The addition of signals is modelled by binary addition. Carry-overs are neglected and the addition is limited to the specific bit-level. The resulting bit-vectors get truncating at the noise bit-level, which yields a deterministic model. Moreover, specific channel gains shift the bit-vectors for a certain number of bit-levels with a shift-matrix $\mathbf{S}$, which is defined as
\begin{equation}
\mathbf{S}=\begin{pmatrix}
0 & 0 &  \cdots & 0 & 0\\
1 & 0 &  \cdots & 0 & 0\\
0 & 1 &  \cdots & 0 & 0\\
\vdots & \vdots & \ddots & \vdots & \vdots \\
0 & 0 &  \cdots & 1 & 0\\
\end{pmatrix}.
\end{equation}
The model therefore shifts an incoming bit vector for $q-n$ positions with $\mathbf{Y}=\mathbf{S}^{q-n}\mathbf{X}$, where $q:=\max\{n\}$.
Channel gains are represented by $n_{ij}$-bit levels which corresponds to $\lceil\log \mbox{SNR}\rceil$ of the original channel. With this definitions, the system model can be approximated by
\begin{IEEEeqnarray}{rCl}
\mathbf{Y}_1&=&\mathbf{S}^{q-n_{11}}\mathbf{X}'_1\oplus\mathbf{S}^{q-n_{21}}\mathbf{X}'_2\IEEEyessubnumber\\
\mathbf{Y}_2&=&\mathbf{S}^{q-n_{22}}\mathbf{X}'_2\oplus\mathbf{S}^{q-n_{12}}\mathbf{X}'_1\IEEEyessubnumber,
\label{LDM_Model}
\end{IEEEeqnarray}
where $q:=\max\{n_{11},n_{12},n_{21},n_{22}\}$. For ease of notation, we denote $\mathbf{X}_1=\mathbf{S}^{q-n_{11}}\mathbf{X}'_1$ and $\mathbf{X}_2=\mathbf{S}^{q-n_{21}}\mathbf{X}'_2$.
Furthermore, we denote $\mathbf{S}^{q-n_{22}}\mathbf{X}'_2$ and $\mathbf{S}^{q-n_{12}}\mathbf{X}'_1$ by $\hat{\mathbf{X}}_2$ and $\hat{\mathbf{X}}_1$, respectively. We also assume that $n_{22}=n_{12}=:n_E$, and denote $n_{1}-n_{2}=:n_\Delta$ with $n_{11}=:n_{1}$ and $n_{21}=:n_2$. We may assume w.l.o.g. that $n_1>n_2$, where we leave out the case that $n_1=n_2$, see remark \eqref{remark2}. To specify a particular range of elements in a bit-level vector we use the notation $\mathbf{a}_{[i:j]}$ to indicate that $\mathbf{a}$ is restricted to the bit-levels $i$ to $j$. Bit levels are counted from top, most significant bit in the expansion, to bottom. If $i=1$, it will be omitted $\mathbf{a}_{[:j]}$, the same for $j\!=\!n$ $\mathbf{a}_{[i:]}$.
\begin{remark}
The assumption that $n_{22}=n_{12}=n_E$, i.e. the eavesdropper receives the signals with equal strength, does not influence the achievable secrecy sum-rate. Consider a channel with $n_{22}\neq n_{12}$, for example $n_{22}>n_{12}$. The part of $\mathbf{X}'_2$ which is received above $n_{12}$ at the eavesdropper, $\hat{\mathbf{X}}_{2,[:n_{22}-n_{12}]}$, cannot be utilized since it cannot be jammed. One can therefore achieve the same rate by ignoring the top $n_{22}-n_{12}$ bits of $\mathbf{X}'_2$. The same argument holds for $n_{12}>n_{22}$.
\end{remark}
\section{Main Result}
\subsection{Achievable Secrecy Rate}
\begin{theorem}
The achievable secrecy sum-rate $R_{\text{ach}}$ of the linear deterministic multiple access wiretap channel with symmetric channel gains at the eavesdropper is
\begin{equation}
R_{\text{ach}}=\begin{cases}
\tfrac{2}{3}(\lfloor \tfrac{n_c}{3n_\Delta} \rfloor3n_\Delta)+n_p+Q. & \text{for } n_2\geq n_E
\\
\tfrac{2}{3}(\lfloor \tfrac{n_c}{3n_\Delta} \rfloor3n_\Delta)+Q.  & \text{for } n_E> n_2,\\
\end{cases}
\end{equation}
where $n_c=n_E+n_\Delta$, $n_p=n_1-n_c$ and 
\begin{equation}
Q=\begin{cases}
q & \text{for } n_Q< n_\Delta 
\\
n_\Delta & \text{for } 2n_\Delta > n_Q \geq  n_\Delta\\
n_\Delta + q & \text{for } n_Q \geq 2n_\Delta,
\end{cases}
\end{equation} with $n_Q= n_c- \lfloor \tfrac{n_c}{3n_\Delta}\rfloor3n_\Delta$ and $q=n_Q-\lfloor \tfrac{n_Q}{n_\Delta} \rfloor n_\Delta$.

\end{theorem}
\subsubsection{Case 1 ($n_2\geq n_E)$}
First of all, we introduce the common and private parts of the received signal $\mathbf{Y}_{1}$.  We count the bits from the top (most-significant bit) downwards. The common part will be denoted by $\mathbf{Y}_{1,c}$ and consist of the top $n_c=n_E+n_\Delta$ bits of $\mathbf{Y}_{1}$. The remainder, the private part, will be denoted as $\mathbf{Y}_{1,p}$ and has $n_1-n_c:=n_p$ bits. It only contains signal parts which are not received by the eavesdropper. Our strategy is to deploy a cooperative jamming scheme such that minimal jamming is done to $\mathbf{Y}_{1,c}$, while maximal jamming is received at $\mathbf{Y}_2$. We denote the part of $\mathbf{X}_{1}$ in $\mathbf{Y}_{1,c}$ by $\mathbf{X}_{1,c}$. Moreover, we denote the part of $\mathbf{X}_2$ which gets received at $\mathbf{Y_2}$ by $\mathbf{X}_{2,c}$. We partition these common signals into $3n_\Delta$-bit parts and partition these parts again into $n_\Delta$-bit parts. For $\mathbf{X}_{1,c}$, in every $3n_\Delta$-bit part we use the first $n_\Delta$ bits for the message and the next $n_\Delta$ bits for jamming, while the last $n_\Delta$ bits will not be used. For $\mathbf{X}_{2,c}$, in every $3n_\Delta$-bit part, the first $n_\Delta$ bits will be used for jamming. The next $n_\Delta$ bits will be used for the message and the last $n_\Delta$ bits left free. There will be a reminder part with 
\begin{equation}
n_Q= n_c- \lfloor \tfrac{n_c}{3n_\Delta}\rfloor3n_\Delta\text{ bits}.
\end{equation}
The reminder part follows the same design rules as the $3n_\Delta$ parts, until $n_Q$ bits are allocated. The scheme is designed such that the jamming parts of $\mathbf{X}_{1,c}$ and $\mathbf{X}_{2,c}$ overlap at $\mathbf{Y}_{1,c}$, while the message parts of one signal overlap with the non-used part of the other signal. However, due to the signal strength difference $n_\Delta$, the jamming parts overlap with the messages at $\mathbf{Y}_{2}$, see Fig.~\ref{Design_Scheme_2}. Secure communication is therefore provided by the crypto-lemma, as long as we use a Bern$(\tfrac{1}{2})$ distribution for the jamming bits. The whole private part can be used for messaging and its sum-rate is therefore $r_p=n_p$.
The achievable secure rate for the common part consists of the rate for the $3n_\Delta$ partitions and the reminder part. It can be seen that every $3n_\Delta$-part of $\mathbf{Y}_{1,c}$ allocates $2n_\Delta$ bits for the messages. This results in the common secrecy rate
\begin{equation}
r_c= (\lfloor \tfrac{n_c}{3n_\Delta} \rfloor3n_\Delta)\tfrac{2}{3} +Q,
\end{equation}
where $Q$ specifies the rate part of the remainder term. In the remainder part we allocate all remaining bits as message bits, as long as $n_Q < n_\Delta$. For $2n_\Delta > n_Q \geq  n_\Delta$, we allocate the first $n_\Delta$ bits of $n_Q$ for the message. And for $n_Q \geq 2n_\Delta$, we allocate the first $n_\Delta$ bits as well as the last $q$ bits, where $q$ is defined as
\begin{equation}
q=n_Q-\lfloor \tfrac{n_Q}{n_\Delta} \rfloor n_\Delta.
\end{equation}
This results in
\begin{equation}
Q=\begin{cases}
q & \text{for } n_Q< n_\Delta 
\\
n_\Delta & \text{for } 2n_\Delta > n_Q \geq  n_\Delta\\
n_\Delta + q & \text{for } n_Q \geq 2n_\Delta.
\end{cases}
\end{equation}
Together with the private rate term, we achieve
\begin{equation}
R=\tfrac{2}{3}(\lfloor \tfrac{n_c}{3n_\Delta} \rfloor3n_\Delta)+n_p+Q.
\end{equation}

\begin{figure}
\includegraphics[scale=1]{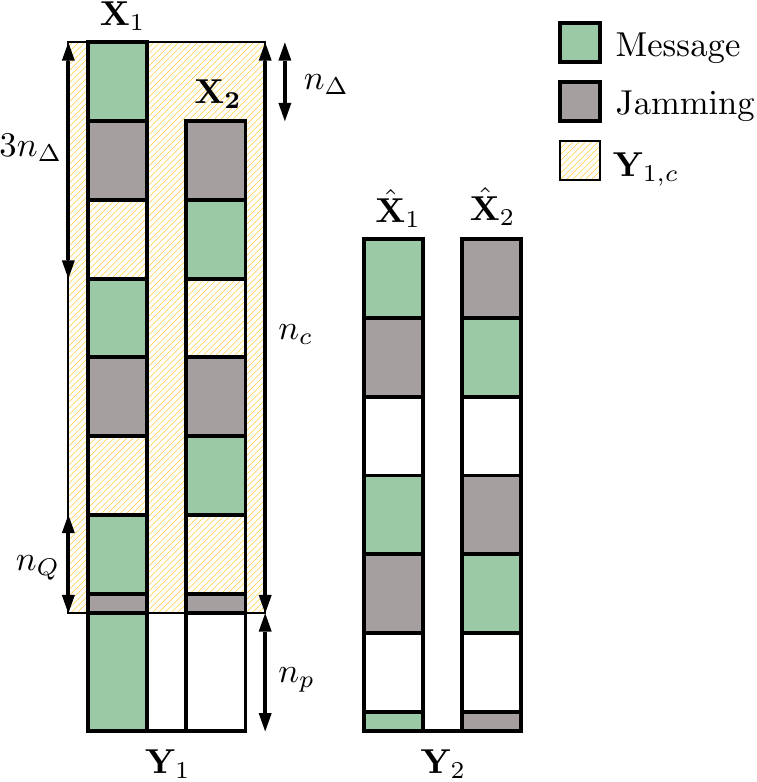}
\caption{Illustration of the achievable scheme. The private part $\mathbf{Y}_{1,p}$ can be used freely and is, in this case, allocated by User 1. The common part $\mathbf{Y}_{1,c}$ uses our alignment strategy. The strategy exploits the channel gain difference between both signals, to minimize the effect of jamming at the receiver $\mathbf{Y}_1$, while jamming all signal parts at the eavesdropper $\mathbf{Y}_2$.}
\label{Design_Scheme_2}
\end{figure}

\subsubsection{Case 2 ($n_2\geq n_E)$}

The achievable scheme is the same as in the previous case, except that we do not have a private part. We therefore have an achievable rate of 
\begin{equation}
R=\tfrac{2}{3}(\lfloor \tfrac{n_c}{3n_\Delta} \rfloor3n_\Delta)+Q.
\end{equation}

\begin{remark}The bit level shift between $\mathbf{Y}_1$ and $\mathbf{Y}_2$ of $n_\Delta$ bits makes it impossible to divide $\mathbf{Y}_1$ in exclusively private and common parts. In our division, the bottom $n_\Delta$ bits of $\mathbf{x}_{1,c}$ are only received at $\mathbf{Y}_1$ and therefore private. Hence, the common rate $r_c$ is not purely made of common signal parts. Nevertheless, our choice of division reaches the upper bound and fits into the scheme.
\label{remark1}
\end{remark}
\begin{remark} Our scheme relies on the signal strength difference between both users. Our scheme would not work, if $n_1=n_2$, while having equal channel gains at the eavesdropper. In that case we would not have any signal strength diversity to exploit which results in a singularity point where the secrecy rate is zero.
\label{remark2}
\end{remark}

\subsection{Converse}
\label{det-converse}
\begin{theorem}
The secrecy sum-rate $R_{\text{ach}}$ of the linear deterministic multiple access wiretap channel with symmetric channel gains at the eavesdropper is bounded from above by
\begin{equation}
r_{UB}=\begin{cases}
\tfrac{2}{3}n_c+n_p+\frac{1}{3}n_\Delta & \text{for } n_2\geq n_E
\\
\tfrac{2}{3}n_c+\frac{1}{3}n_\Delta  & \text{for } n_E> n_2.\\
\end{cases}
\end{equation}

\end{theorem}

\begin{proof}We start with some general observations and derivations before handling the different cases explicitly. We begin with the following derivations
\begin{IEEEeqnarray*}{rCl}
n(R_1+R_2)&=& H(W_1,W_2)\\
&=& H(W_1,W_2|\mathbf{Y}_1^n)+I(W_1,W_2;\mathbf{Y}_1^n)\\
&\leq & I(W_1,W_2;\mathbf{Y}_1^n)+ n\epsilon\\
&\overset{(a)}{\leq} & I(W_1,W_2;\mathbf{Y}_1^n)-I(W_1,W_2;\mathbf{Y}_2^n)+n\epsilon_2\\
&\leq & I(W_1,W_2;\mathbf{Y}_1^n,\mathbf{Y}_2^n)-I(W_1,W_2;\mathbf{Y}_2^n)+n\epsilon_2\\
&\leq & I(W_1,W_2;\mathbf{Y}_1^n|\mathbf{Y}_2^n)+n\epsilon_2\\
&\leq & I(\mathbf{X}_1^n,\mathbf{X}_2^n;\mathbf{Y}_1^n|\mathbf{Y}_2^n)+n\epsilon_2\\
&=&H(\mathbf{Y}_1^n|\mathbf{Y}_2^n)
	-H(\mathbf{Y}_1^n|\mathbf{Y}_2^n,\mathbf{X}_1^n,\mathbf{X}_2^n)+n\epsilon_2\\
&\overset{(b)}{=}&H(\mathbf{Y}_1^n|\mathbf{Y}_2^n)+n\epsilon_2\\
&\leq & H(\mathbf{Y}_{1,c}^n|\mathbf{Y}_{2}^n)+H(\mathbf{Y}_{1,p}^n|\mathbf{Y}_{2}^n,\mathbf{Y}_{1,c}^n)+n\epsilon_2\IEEEyesnumber
\label{UB:EQ1}
\end{IEEEeqnarray*} where we used basic techniques such as Fano's inequality and the chain rule. Step $(a)$ introduces the secrecy constraint \eqref{security_constraint}, while we used  the chain rule, non-negativity of mutual information and the data processing inequality in the following lines. Step $(b)$ follows from the fact that $\mathbf{Y}_1^n$ is a function of $(\mathbf{X}_1^n,\mathbf{X}_2^n)$. Note that due to the definition of the common and the private part\footnote{The common part is defined as $\mathbf{Y}_{1,c}^n=\mathbf{Y}_{1,[:n_c]}^n$, and the private part as $\mathbf{Y}_{1,p}^n=\mathbf{Y}_{1,[n_c+1:]}^n$.} of $\mathbf{Y}_1^n$, it follows that $H(\mathbf{Y}_{1,p}^n|\mathbf{Y}_{2}^n,\mathbf{Y}_{1,c}^n)=0$ for $n_E\geq n_2$.

We now extend the strategy of \cite{XieUlukusOneHop}, of bounding a single signal part, to asymmetrical channel gains
\begin{IEEEeqnarray*}{rCl}
nR_1 &=&H(W_1)\\
&\leq & I(W_1;\mathbf{Y}_1^n)-n \epsilon_3\\
&\leq & I(\mathbf{X}_1^n;\mathbf{Y}_1^n)-n \epsilon_3\\
&=& I(\mathbf{X}_1^n;\mathbf{Y}_{1,c}^n)+I(\mathbf{X}_1^n;\mathbf{Y}_{1,p}^n|\mathbf{Y}_{1,c}^n)-n \epsilon_3\\
&=& H(\mathbf{Y}_{1,c}^n)-H(\mathbf{Y}_{1,c}^n|\mathbf{X}_1^n)+I(\mathbf{X}_1^n;\mathbf{Y}_{1,p}^n|\mathbf{Y}_{1,c}^n)-n \epsilon_3\\
&=&H(\mathbf{Y}_{1,c}^n)-H(\mathbf{X}_{2,c}^n)+I(\mathbf{X}_1^n;\mathbf{Y}_{1,p}^n|\mathbf{Y}_{1,c}^n)-n \epsilon_3
\label{UB:EQ3}
\end{IEEEeqnarray*}
and it therefore holds that
\begin{equation}
H(\mathbf{X}_{2,c}^n) \leq H(\mathbf{Y}_{1,c}^n)+I(\mathbf{X}_1^n;\mathbf{Y}_{1,p}^n|\mathbf{Y}_{1,c}^n)-n(R_1+\epsilon_3).
\label{UB:EQ4}
\end{equation}

The same can be shown for $H(\mathbf{X}_{1,c}^n)$, where it holds that
\begin{equation}
H(\mathbf{X}_{1,c}^n) \leq H(\mathbf{Y}_{1,c}^n)+I(\mathbf{X}_2^n;\mathbf{Y}_{1,p}^n|\mathbf{Y}_{1,c}^n)-n(R_2+\epsilon_4).
\label{UB:EQ5}
\end{equation}
Moreover, we have that
\begin{IEEEeqnarray*}{rCl}
&&I(\mathbf{X}_1^n;\mathbf{Y}_{1,p}^n|\mathbf{Y}_{1,c}^n)+I(\mathbf{X}_2^n;\mathbf{Y}_{1,p}^n|\mathbf{Y}_{1,c}^n)\\
&=&\:2H(\mathbf{Y}_{1,p}^n|\mathbf{Y}_{1,c}^n)-H(\mathbf{Y}_{1,p}^n|\mathbf{Y}_{1,c}^n,\mathbf{X}_1^n)-H(\mathbf{Y}_{1,p}^n|\mathbf{Y}_{1,c}^n,\mathbf{X}_2^n)\\
&=& 2H(\mathbf{Y}_{1,p}^n|\mathbf{Y}_{1,c}^n)-\:H(\mathbf{X}_{2,p}^n|\mathbf{Y}_{1,c}^n)-H(\mathbf{X}_{1,p}^n|\mathbf{Y}_{1,c}^n)\\
&= &\:H(\mathbf{Y}_{1,p}^n|\mathbf{Y}_{1,c}^n)\IEEEyesnumber.
\label{UB:EQ10}
\end{IEEEeqnarray*}
The key idea for the various cases is now to bound the term $H(\mathbf{Y}_{1,c}^n|\mathbf{Y}_{2}^n)$, or equivalently $H(\mathbf{Y}_{1}^n|\mathbf{Y}_{2}^n)$ for $n_E>n_2$, in an appropriate way, to be able to use \eqref{UB:EQ4} and \eqref{UB:EQ5} on \eqref{UB:EQ1}. 
We start with the first case:
\subsubsection{Case 1 ($n_2\geq n_E$)}
Here we have a none vanishing private part, due to the definition of $\mathbf{Y}_{1,c}^n$ and therefore need to bound the term $H(\mathbf{Y}_{1,c}^n|\mathbf{Y}_{2}^n)$. Note that due to the definition of $\mathbf{Y}_{1,c}^n$ we have that $H(\mathbf{X}_{2,c}^n)=H(\hat{\mathbf{X}}_{2}^n)$. We look into the first term of equation \eqref{UB:EQ1} and show that
\begin{IEEEeqnarray*}{rCl}
H(\mathbf{Y}_{1,c}^n|\mathbf{Y}_{2}^n) &=& H(\mathbf{Y}_{1,c}^n,\mathbf{Y}_{2}^n)-H(\mathbf{Y}_{2}^n)\\
&\leq & H(\mathbf{Y}_{1,c}^n,\hat{\mathbf{X}}_1^n,\hat{\mathbf{X}}_2^n)-H(\mathbf{Y}_{2}^n)\\
&=& H(\hat{\mathbf{X}}_1^n,\hat{\mathbf{X}}_2^n)+H(\mathbf{Y}_{1,c}^n|\hat{\mathbf{X}}_1^n,\hat{\mathbf{X}}_2^n)-H(\mathbf{Y}_{2}^n)\\
& \leq &H(\hat{\mathbf{X}}_1^n)+H(\hat{\mathbf{X}}_2^n)-H(\mathbf{Y}_{2}^n|\hat{\mathbf{X}}_2^n)\\
&&+\:H(\mathbf{Y}_{1,c}^n|\hat{\mathbf{X}}_1^n,\hat{\mathbf{X}}_2^n)\\
&=& H(\hat{\mathbf{X}}_2^n)+H(\mathbf{Y}_{1,c}^n|\hat{\mathbf{X}}_1^n,\hat{\mathbf{X}}_2^n).\IEEEyesnumber
\label{UB:EQ2}
\end{IEEEeqnarray*}
Observe that the second term of equation \eqref{UB:EQ2} is depended on the specific regime. We can bound this term by
\begin{equation}
H(\mathbf{Y}_{1,c}^n|\hat{\mathbf{X}}_1^n,\hat{\mathbf{X}}_2^n)\leq n(n_c-n_E)= nn_\Delta.
\end{equation}

Note that the choice of $\hat{\mathbf{X}}_2^n$ in \eqref{UB:EQ2} as remaining signal part was arbitrary due to our assumption that both signals $\hat{\mathbf{X}}_1^n$ and $\hat{\mathbf{X}}_2^n$ have the same signal strength. Moreover, it follows on the same lines that
\begin{equation}
H(\mathbf{Y}_{1,c}^n|\mathbf{Y}_{2}^n)\leq H(\hat{\mathbf{X}}_1^n)+nn_\Delta.
\label{UB:EQ6}
\end{equation}
Looking at this result, its intuitive that one can also show the stronger result
\begin{equation}
H(\mathbf{Y}_{1,c}^n|\mathbf{Y}_{2}^n)\leq H(\mathbf{X}_{1,c}^n)
\label{UB:EQ7}
\end{equation}
for the case that $n_2 \geq n_E$.
This can be shown by considering a similar strategy as in \eqref{UB:EQ2} 
\begin{IEEEeqnarray*}{rCl}
H(\mathbf{Y}_{1,c}^n|\mathbf{Y}_{2}^n) &=& H(\mathbf{Y}_{1,c}^n,\mathbf{Y}_{2}^n)-H(\mathbf{Y}_{2}^n)\\
&\leq & H(\mathbf{Y}_{2}^n,\mathbf{X}_{1,c}^n,\mathbf{X}_{2,c}^n)-H(\mathbf{Y}_{2}^n)\\
&=& H(\mathbf{X}_{1,c}^n,\mathbf{X}_{2,c}^n)+H(\mathbf{Y}_{2}^n|\mathbf{X}_{1,c}^n,\mathbf{X}_{2,c}^n)-H(\mathbf{Y}_{2}^n)\\
& \leq &H(\mathbf{X}_{1,c}^n)+H(\mathbf{X}_{2,c}^n)-H(\mathbf{Y}_{2}^n|\mathbf{X}_{1,c}^n)\\
&&+\:H(\mathbf{Y}_{2}^n|\mathbf{X}_{1,c}^n,\mathbf{X}_{2,c}^n)-H(\mathbf{Y}_{2}^n)\\
&=& H(\mathbf{X}_{1,c}^n)+H(\mathbf{Y}_{2}^n|\mathbf{X}_{1,c}^n,\mathbf{X}_{2,c}^n),\IEEEyesnumber
\label{UB:EQ8}
\end{IEEEeqnarray*} where 
\begin{equation}
H(\mathbf{Y}_{2}^n|\mathbf{X}_{1,c}^n,\mathbf{X}_{2,c}^n)\leq (n_E-n_2)^+=0.
\label{UB:EQ9}
\end{equation}

We combine one sum-rate inequality \eqref{UB:EQ1} with \eqref{UB:EQ2} and one with \eqref{UB:EQ8}.
Moreover, we plug \eqref{UB:EQ4} and \eqref{UB:EQ5} into the corresponding bound, which yields

\begin{IEEEeqnarray*}{rCl}
n(2R_1+R_2 -\epsilon_6) &\leq & H(\mathbf{Y}_{1,c}^n)+I(\mathbf{X}_2^n;\mathbf{Y}_{1,p}^n|\mathbf{Y}_{1,c}^n)\\
&&+\:H(\mathbf{Y}_{1,p}^n|\mathbf{Y}_{2}^n,\mathbf{Y}_{1,c}^n)
\end{IEEEeqnarray*} and
\begin{IEEEeqnarray*}{rCl}
n(R_1+2R_2 -\epsilon_7) &\leq & H(\mathbf{Y}_{1,c}^n)+I(\mathbf{X}_1^n;\mathbf{Y}_{1,p}^n|\mathbf{Y}_{1,c}^n)\\
&&+\:H(\mathbf{Y}_{1,p}^n|\mathbf{Y}_{2}^n,\mathbf{Y}_{1,c}^n)+nn_\Delta.
\end{IEEEeqnarray*}
A summation of these results gives
 \begin{IEEEeqnarray*}{rCl}
3n(R_1+R_2) -n\epsilon_8 &\leq & 2H(\mathbf{Y}_{1,c}^n)+I(\mathbf{X}_1^n;\mathbf{Y}_{1,p}^n|\mathbf{Y}_{1,c}^n)\\
&&+\:I(\mathbf{X}_2^n;\mathbf{Y}_{1,p}^n|\mathbf{Y}_{1,c}^n)+nn_\Delta\\
&&+\:2H(\mathbf{Y}_{1,p}^n|\mathbf{Y}_{2}^n,\mathbf{Y}_{1,c}^n).
\end{IEEEeqnarray*}
Using \eqref{UB:EQ10}, and the fact that $H(\mathbf{Y}_{1,p}^n|\mathbf{Y}_{2}^n,\mathbf{Y}_{1,c}^n)\leq nn_p$ and $H(\mathbf{Y}_{1,p}^n|\mathbf{Y}_{1,c}^n)\leq nn_p$ results in 
\begin{IEEEeqnarray*}{rCl}
3n(R_1+R_2) -n\epsilon_8 &\leq & 2H(\mathbf{Y}_{1,c}^n)+3nn_p+nn_\Delta\\
&\leq & 2nn_c+3nn_p+nn_\Delta.
\end{IEEEeqnarray*}
Dividing by $3n$ and letting $n\rightarrow\infty$ shows the result.

\subsubsection{Case 2 ($n_E>n_2$)}
First, we assume that $n_E \geq n_1$, and include a short proof for $n_1>n_E\geq n_2$ at the end of this subsection. For Case 2, the private part $\mathbf{Y}_{1,p}$ is zero, due to the definition of the private part and $n_E>n_2$. It follows that \eqref{UB:EQ1} is
\begin{equation}
n(R_1+R_2)\leq H(\mathbf{Y}_1^n|\mathbf{Y}_2^n).
\label{UB:EQ11}
\end{equation}
Moreover, $H(\mathbf{X}_2^n)=H(\mathbf{X}_{2,c}^n)\leq H(\hat{\mathbf{X}}_2^n)$, which is why we need to bound \eqref{UB:EQ11} by $H(\mathbf{X}_2^n)$ and $H(\mathbf{X}_1^n)$.
We therefore modify \eqref{UB:EQ8} to fit our purpose in the following way
\begin{IEEEeqnarray*}{rCl}
H(\mathbf{Y}_{1}^n|\mathbf{Y}_{2}^n) &=& H(\mathbf{Y}_{1}^n,\mathbf{Y}_{2}^n)-H(\mathbf{Y}_{2}^n)\\
&\leq & H(\mathbf{X}_{1}^n,\mathbf{X}_{2}^n)+H(\mathbf{Y}_{2}^n|\mathbf{X}_{1}^n,\mathbf{X}_{2}^n)-H(\mathbf{Y}_{2}^n)\\
&=& H(\mathbf{X}_{1}^n,\mathbf{X}_{2}^n)+H(\mathbf{Y}_{2,c}^n|\mathbf{X}_{1}^n,\mathbf{X}_{2}^n)\\
&&+\:H(\mathbf{Y}_{2,p}^n|\mathbf{X}_{1}^n,\mathbf{X}_{2}^n,\mathbf{Y}_{2,c}^n)\\
&&-\:H(\mathbf{Y}_{2,c}^n)-H(\mathbf{Y}_{2,p}^n|\mathbf{Y}_{2,c}^n)\\
&\leq &H(\mathbf{X}_{1}^n,\mathbf{X}_{2}^n)+H(\mathbf{Y}_{2,c}^n|\mathbf{X}_{1}^n,\mathbf{X}_{2}^n)-H(\mathbf{Y}_{2,c}^n),
\label{UB:EQ12}
\end{IEEEeqnarray*}
where $\mathbf{Y}_{2,c}^n=\mathbf{Y}_{2,[:n_1]}^n$ and $\mathbf{Y}_{2,p}^n=\mathbf{Y}_{2,[n_1+1:]}^n$.
Now, we can show that
\begin{IEEEeqnarray*}{rCl}
H(\mathbf{Y}_{1}^n|\mathbf{Y}_{2}^n)& \leq &H(\mathbf{X}_{1}^n,\mathbf{X}_{2}^n)+H(\mathbf{Y}_{2,c}^n|\mathbf{X}_{1}^n,\mathbf{X}_{2}^n)-H(\mathbf{Y}_{2,c}^n)\\
&=&H(\mathbf{X}_{1}^n,\mathbf{X}_{2}^n)+H(\hat{\mathbf{X}}_{2}^n|\mathbf{X}_{2}^n)-H(\mathbf{Y}_{2,c}^n)\\
&\leq &H(\mathbf{X}_{1}^n)+H(\mathbf{X}_{2}^n)-H(\mathbf{Y}_{2,c}^n|\mathbf{X}_{2}^n)\\
&&+\:H(\hat{\mathbf{X}}_{2}^n|\mathbf{X}_{2}^n)\\
&=& H(\mathbf{X}_{2}^n)+H(\hat{\mathbf{X}}_{2}^n|\mathbf{X}_{2}^n)\\
&\leq &H(\mathbf{X}_{2}^n)+ nn_\Delta.\IEEEyesnumber
\label{UB:EQ13}
\end{IEEEeqnarray*}
Bounding $H(\mathbf{Y}_{1}^n|\mathbf{Y}_{2}^n)$ by $H(\mathbf{X}_{1}^n)$ requires more work. We have a redundancy in the negative entropy terms, with which we can cancel the $H(\hat{\mathbf{X}}_{2}^n|\mathbf{X}_{2}^n)$ term in the following way
\begin{IEEEeqnarray*}{rCl}
H(\mathbf{Y}_{1}^n|\mathbf{Y}_{2}^n)& \leq &H(\mathbf{X}_{1}^n,\mathbf{X}_{2}^n)+H(\hat{\mathbf{X}}_{2}^n|\mathbf{X}_{2}^n)-H(\mathbf{Y}_{2,c}^n)\\
&\leq &H(\mathbf{X}_{1}^n)+H(\mathbf{X}_{2}^n)-H(\mathbf{Y}_{2,c,[:n_2]}^n|\mathbf{X}_{1}^n)\\
&&+\:H(\hat{\mathbf{X}}_{2}^n|\mathbf{X}_{2}^n)-H(\mathbf{Y}_{2,c,[n_2+1:]}^n|\mathbf{X}_{1}^n,\mathbf{Y}_{2,c,[:n_2]}^n)\\
&=& H(\mathbf{X}_{1}^n)-H(\mathbf{Y}_{2,c,[n_2+1:]}^n|\mathbf{X}_{1}^n,\mathbf{Y}_{2,c,[:n_2]}^n)\\
&&+\:H(\hat{\mathbf{X}}_{2}^n|\mathbf{X}_{2}^n)\\
&\leq &H(\mathbf{X}_{1}^n)-H(\mathbf{Y}_{2,c,[n_2+1:]}^n|\mathbf{X}_{1}^n,\mathbf{Y}_{2,c,[:n_2]}^n,\mathbf{X}_{2}^n)\\
&&+\:H(\hat{\mathbf{X}}_{2}^n|\mathbf{X}_{2}^n)\\
&=&H(\mathbf{X}_{1}^n)-H(\mathbf{Y}_{2,c,[n_2+1:]}^n|\mathbf{X}_{1}^n,\mathbf{X}_{2}^n)\\
&&+\:H(\hat{\mathbf{X}}_{2}^n|\mathbf{X}_{2}^n)\\
&=& H(\mathbf{X}_{1}^n)-H(\hat{\mathbf{X}}_{2,c,[n_2+1:]}^n|\mathbf{X}_{2}^n)+H(\hat{\mathbf{X}}_{2}^n|\mathbf{X}_{2}^n)\\
&= &H(\mathbf{X}_{1}^n).\IEEEyesnumber
\label{UB:EQ14}
\end{IEEEeqnarray*}
Now we can bound one \eqref{UB:EQ11} with \eqref{UB:EQ13} and one with \eqref{UB:EQ14}. Moreover, we use \eqref{UB:EQ4} and \eqref{UB:EQ5} on the result. Note that due to our regime, \eqref{UB:EQ4} becomes
\begin{equation}
H(\mathbf{X}_{2}^n) \leq H(\mathbf{Y}_{1}^n)-n(R_1+\epsilon_3),
\end{equation}
while \eqref{UB:EQ5} becomes
\begin{equation}
H(\mathbf{X}_{1}^n) \leq H(\mathbf{Y}_{1}^n)-n(R_2+\epsilon_4).
\end{equation}
Putting everything together results in
 \begin{IEEEeqnarray*}{rCl}
3n(R_1+R_2)-n\epsilon_8 &\leq & 2H(\mathbf{Y}_{1}^n)+nn_\Delta\\
&\leq & 2nn_c+nn_\Delta.
\end{IEEEeqnarray*}
Dividing by $3n$ and letting $n\rightarrow\infty$ shows the result.

We need to modify a bound on $H(\mathbf{Y}_{1}^n|\mathbf{Y}_{2}^n)$, if the signal strength $n_E$ lies in between $n_1$ and $n_2$. In \eqref{UB:EQ13}, we see that 
\begin{equation}
H(\mathbf{X}_{1}^n)-H(\mathbf{Y}_{2,c}^n|\mathbf{X}_{2}^n)\leq n(n_1-n_E)^+.
\end{equation}
Moreover, we have that $H(\hat{\mathbf{X}}_{2}^n|\mathbf{X}_{2}^n)\leq n(n_E-n_2)^+$. Both changes cancel and we get the same result as \eqref{UB:EQ13}. The result follows on the same lines as in the previous derivation.
\end{proof}

\section{Conclusion}
\begin{figure}
\includegraphics[scale=0.3]{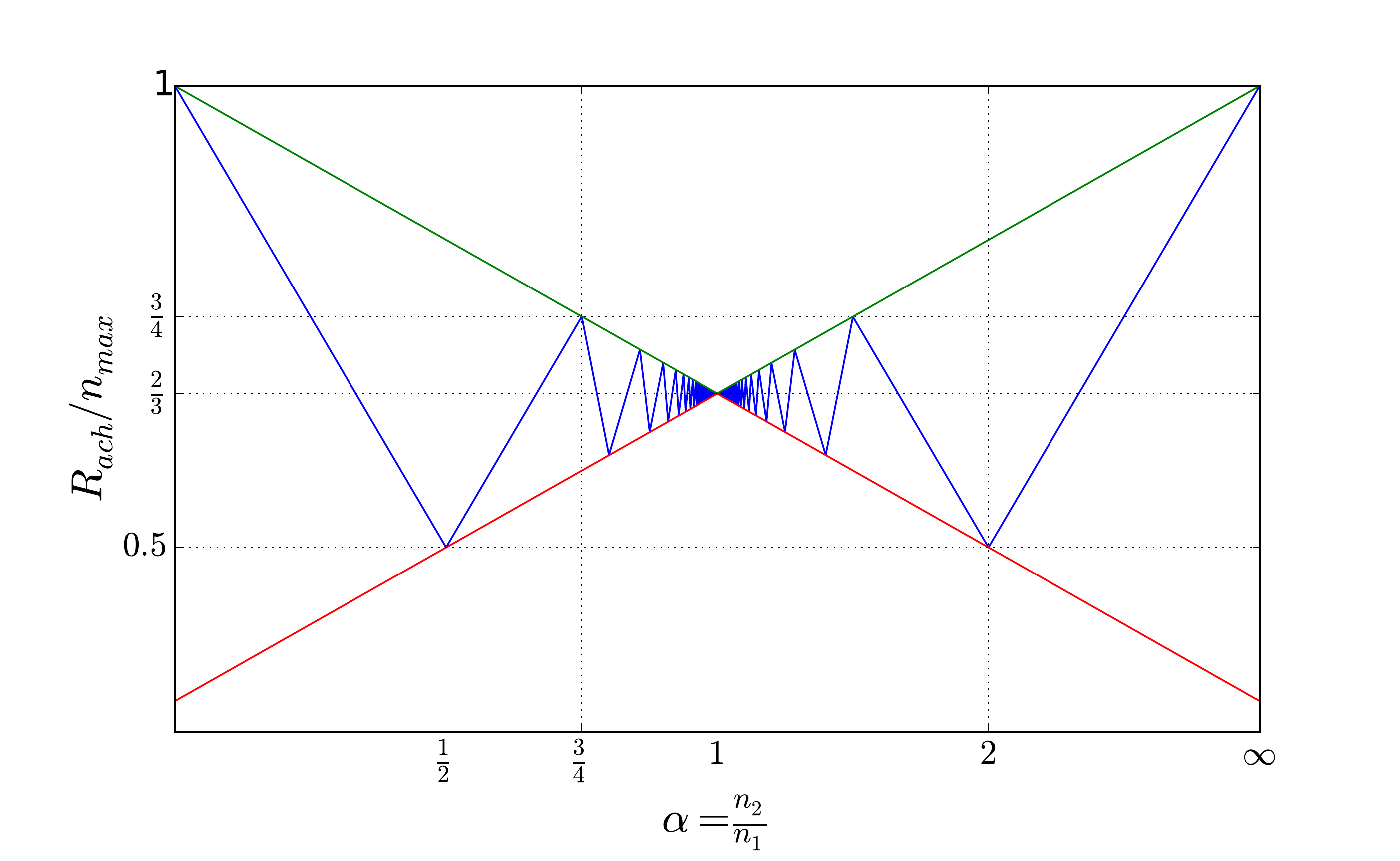}
\caption{The graphic shows the achievable secrecy rate, normalized by the strongest signal $n_{\text{max}}=\max\{n_1,n_2\}$, in blue. Moreover, it shows the upper bound $r_{UB}$ in green. The regime is $n_E>\min\{n_1,n_2\}$, which results in the private rate being zero. The x-axis visualizes the fraction $\alpha=\tfrac{n_2}{n_1}$. At $\alpha=0$, there is no penalty in utilizing $\mathbf{X}_2$ for jamming, and $n_{\text{max}}$ bits can be achieved. For $\alpha\rightarrow 1$, the difference between $\mathbf{X_1}$ and $\mathbf{X}_2$ gets smaller and the rate approaches the s.d.o.f of $\tfrac{2}{3}n_{\text{max}}$. For $1<\alpha\leq \infty$, $\mathbf{X}'_2$ gets stronger than $\mathbf{X}'_1$. Interchanging the roles of both users in the scheme results in a mirror-symmetric behaviour. We also show a red curve for $\tfrac{2}{3}\tfrac{n_1+n_2}{2}$.}
\label{Fig3}
\end{figure}

\subsection{Discussion of the Results}
We have approximated the Gaussian multiple access wiretap channel with the linear deterministic model. This enables a simplified look at the model with an emphasis on the role of channels gains in the secrecy sum-rate analysis. We used a signal-scale alignment approach to provide information-theoretic security. Our approach distinguishes between channel gain differences which is reflected in the achievable sum-rate. Our results agree with previous s.d.o.f. results, as our secrecy sum-rate approaches the s.d.o.f. asymptotically for $n_\Delta \rightarrow 0$ and $n_E\geq \min\{n_2,n_1\}$, i.e. without channel gain difference between users and without private part, see Fig.~\ref{Fig3}. Moreover, we have shown upper bounds which are tight for certain $n_\Delta$ ranges. We note that the achievable sum-rate varies between being above and below the $\tfrac{2}{3}$ threshold. An interesting question is, to which channel gain the $\tfrac{2}{3}$ d.o.f. refer to. The best d.o.f. representation in our scheme is $\tfrac{2}{3}$ of $\frac{n_1+n_2}{2}$, which is represented in the red curve in Fig.~\ref{Fig3}.
\subsection{Future Research Directions}

We think that the results of this paper are important for a constant-gap sum-capacity analysis of the Gaussian wiretap channel. The achievable scheme can be translated to the Gaussian equivalent by using layered lattices codes. The signals have to be partitioned into $n_\Delta$ intervals, where every partition represents a single-link using one lattice code. Moreover, our upper bounds guide the design of smart genies for the Gaussian channel. We expect that the techniques of \cite{XieUlukusOneHop} together with smart genies, which partition the terms accordingly, will result in bounds which lead to a constant-gap result for the Gaussian multiple access wiretap channel. 

\bibliographystyle{./IEEEtran}
\bibliography{./ref}

\end{document}